\documentclass[sigconf]{acmart}

\usepackage{caption}
\usepackage{subcaption}
\usepackage{boldline}
\usepackage{multirow}
\usepackage{makecell}
\usepackage{ctable}
\usepackage{booktabs}
\usepackage{hhline}
\usepackage{cellspace}

\usepackage[font=small,skip=1pt]{caption}
\usepackage[compact]{titlesec}
\titlespacing{\section}{0pt}{2ex}{1ex}
\titlespacing{\subsection}{0pt}{1ex}{0ex}
\titlespacing{\subsubsection}{0pt}{0.5ex}{0ex}

\def\bw{\mathbf{w}}
\def\bz{\mathbf{z}}

\AtBeginDocument{%
  \providecommand\BibTeX{{%
    \normalfont B\kern-0.5em{\scshape i\kern-0.25em b}\kern-0.8em\TeX}}}

\copyrightyear{2020} 
\acmYear{2020} 
\setcopyright{acmlicensed}
\acmConference[ICAIF '20]{ACM International Conference on AI in Finance}{October 15--16, 2020}{New York, NY, USA}
\acmBooktitle{ACM International Conference on AI in Finance (ICAIF '20), October 15--16, 2020, New York, NY, USA}
\acmPrice{15.00}
\acmDOI{10.1145/3383455.3422557}
\acmISBN{978-1-4503-7584-9/20/10}

\begin{document}
\title{Choosing News Topics to Explain Stock Market Returns}

\author{Paul Glasserman}
\authornote{Authors are listed in alphabetical order.}
\email{pg20@columbia.edu}
\affiliation{%
  \institution{Columbia Business School}
}

\author{Kriste Krstovski}
\email{kriste.krstovski@columbia.edu}
\affiliation{%
  \institution{Columbia Business School}
}

\author{Paul Laliberte}
\email{pl2669@columbia.edu}
\affiliation{%
  \institution{Columbia Business School}
}

\author{Harry Mamaysky}
\email{hm2646@columbia.edu}
\affiliation{%
  \institution{Columbia Business School}
}

\begin{abstract}
We analyze methods for selecting topics in news articles to explain stock returns. We find, through empirical and theoretical results, that supervised Latent Dirichlet Allocation (sLDA) implemented through Gibbs sampling in a stochastic EM algorithm will often overfit returns to the detriment of the topic model. We obtain better out-of-sample performance through a random search of plain LDA models. A branching procedure that reinforces effective topic assignments often performs best. We test methods on an archive of over 90,000 news articles about S\&P 500 firms.
\end{abstract}

\ccsdesc[500]{Computing methodologies~Information extraction}
\ccsdesc[500]{Computing methodologies~Feature selection}

\keywords{text analysis, finance, supervised topic models}

\maketitle

\section{Introduction}
\label{s:intro}

News sentiment has been widely observed to help forecast stock market returns. The ability of news \emph{topics} to {\em explain} contemporaneous returns has received less attention. Topics are relevant because, for example, news of a fraud investigation at a company may explain a negative return, while a new product announcement will often be associated with positive returns. Topics may also be relevant to explaining return volatility.  Furthermore, topics that are useful for explaining contemporaneous returns may also be useful for forecasting future returns.
For example, a 2\% stock price drop associated with a topic about accounting scandals may be associated with continued price drops in the future; whereas a 2\% stock price drop associated with a topic about a large hedge fund liquidation may be expected to reverse over the short-term.  Not all 2\% price drops are created equal, and being able to associate a particular move to an underlying reason may yield important insights about what comes next.  This paper fits into a broader research agenda that aims to first explain {\em  contemporaneous} price moves, and only then to think about forecasting {\em future} ones.
See \cite{tetlock2014information} and references there for background on news and text analysis in financial economics.

With these considerations in mind, we are interested in finding topics in news text specifically designed to explain market responses. This is a problem of supervised topic modeling, in which a news article about a company is labeled with the company's contemporaneous stock return on the day the news article is released.

We study the use of the supervised Latent Dirichlet Allocation (sLDA) framework of \cite{mcauliffe_blei_2008}. In the sLDA generative model, a document's label is generated through a function of the proportions of words in the document associated with each topic. The topic model is otherwise as in plain LDA \cite{blei_ea_2003}. In most applications of sLDA, the labels are discrete and the objective is to classify documents. Our goal is to explain contemporaneous returns (or volatility), so we face a regression problem rather than a classification task. We also sketch a foundation for this formulation based on finance theory.

SLDA can be implemented through a combination of Gibbs sampling and a stochastic EM algorithm, as in the R package \texttt{lda} and \cite{boydgraber_resnik_2010}.  The ``M'' step updates coefficients by regressing returns (labels) on the topic distribution in each document. The ``E'' step updates topic assignments conditional on a document's words and label.

We show that the performance of the algorithm is very sensitive to the relative weight applied to the text and label in the ``E'' step, where the relative weight is controlled by a variance parameter in the regression model. With too much weight on the labels, the ``E'' step will often drive the algorithm to overfit returns to the detriment of the topic model, resulting in poor out-of-sample performance in explaining returns. We provide some theoretical explanation for this behavior. With too little weight on the labels, sLDA tends to find the same topics as plain LDA.

To get around this difficulty, we propose a random search of alternative LDA estimates that reinforces topic assignments that help explain returns. In our branching method, we start with multiple independent runs of LDA through Gibbs sampling. Of these, we select the one that yields the best explanatory power for returns. This winner becomes the starting point for a new set of LDA runs, and the process repeats. Crucially, in each round the winner is selected through cross-validation. This ensures that our selection process favors topic assignments that explain returns (out of sample) while avoiding the overfitting that results from the in-sample estimates used in the ``E'' step of the stochastic EM algorithm.  

We compare this branching algorithm (and some simpler alternatives) with plain LDA and sLDA. In regressions of returns on topic assignments of a hold-out sample, we find that the search methods improve over plain LDA and significantly outperform sLDA.  Stock market returns are of course very difficult to explain \cite{roll-1988}, so the absolute levels of $R^2$ in our regressions are very small, but the relative improvements using the methods we propose are substantial.

To get a sense of the magnitudes, consider the analysis in \cite{tsm-2008}.   Using over 20 years of stock return data for S\&P 500 firms, they find a regression with a rich set of forecasting variables (lagged prices, accounting information, and text-based measures of news flow) has an $R^2$ for next day returns of roughly 0.2\%.  A similar regression for {\em same day} returns has an $R^2$ of 0.45\%.  Both are in-sample numbers.  Our topic model's out-of-sample ability to explain contemporaneous returns far exceeds the 0.45\% threshold from the prior literature, and may, in fact, represent a sort of natural bound on how much of contemporaneous return variation is explainable in the first place, with the remainder being inexplicable noise.%
\footnote{\cite{nguyen-shirai-2015-topic} uses an LDA-type model to forecast whether the next day's stock move will be up or down -- a very different task than obtaining an expected return from a regression -- and finds that, for a sample of five stocks over the course of one year, discussion topics on Yahoo Finance Message Board lead to a 56\% correct classification rate.  Given the small sample size of the analysis, the statistical significance of this finding is not clear.}

\section{Supervised Topic Models}
\label{s:stm}

This section discusses topic models, from plain LDA to sLDA.

\subsection{Latent Dirichlet Allocation}
\label{s:lda}

A corpus is a collection of documents, and a document is a collection of words. A topic is a probability distribution over words in a corpus. For example, one topic might assign relatively high probability to the words ``debt,'' ``credit,'' and ``rating,'' and another topic to the terms ``China,'' ``Japan,'' and ``Asia.'' We observe words but not the identity of their associated topics. The word ``bond,'' might result from a debt markets topic or from an Asian markets topic. To infer the latent topic assignments from the observed text, we need a model of how text is generated.

Let $\bw = \{\bw(i,j), i=1\dots,N_j,j=1,\dots,J\}$ denote a corpus of $J$ documents, in which $\bw(i,j)$ is the $i$th word in document $j$, and $N_j$ is the total number of words in document $j$. Words are drawn from a collection wide vocabulary of size $V$; we identify the vocabulary with the set $\{1,\dots,V\}$ and take each word $\bw(i,j)$ to be an element of this set. Suppose the text is generated by $K$ topics. Let $\bz = \{\bz(i,j), i=1\dots,N_j,j=1,\dots,J\}$, where $\bz(i,j)\in \{1,\dots,K\}$ is the topic assignment for $\bw(i,j)$.

LDA posits that the corpus is generated as follows:
\begin{enumerate}
\item For each topic $k=1,\dots,K$, generate a random probability distribution $\phi_k$ over the vocabulary $\{1,\dots,V\}$;
\item For each document $j=1,\dots,J$, generate a random probability distribution $\theta_j$ over the set of topics $\{1,\dots,K\}$;
\begin{enumerate}
\item For each $i=1,\dots,N_j$,
\begin{enumerate}
\item draw topic $\bz(i,j)$ from the distribution $\theta_j$
\item draw word $\bw(i,j)$ from the distribution $\phi_{\bz(i,j)}$
\end{enumerate}
\end{enumerate}
\end{enumerate}
In Step (1), each $\phi_k$ is drawn from a Dirichlet distribution with parameters $(\beta_1,\dots,\beta_V)$, and in Step (2), each $\theta_j$ is drawn from a Dirichlet distribution with parameters $(\alpha_1,\dots,\alpha_K)$. (A Dirichlet distribution is a distribution over nonnegative vectors with components that sum to 1.) We will make the standard assumptions that $\beta_1=\cdots=\beta_V=\beta$, $\alpha_1=\cdots=\alpha_K=\alpha$, and that these hyperparameters are known along with $K$. In our implementation, we use the standard values of $\alpha=5/K$ and $\beta=0.01$. 

The LDA generative model implicitly determines the joint distribution
$p(\bz,\bw)$ of topic assignments and words. Only the words
$\bw$ are observed, so inference in LDA entails estimating or
approximating the posterior distribution $p(\bz|\bw)$ of topic
assignments given the text. From $\bz$, one can estimate
the topic distributions $\phi_k$ and the prevalence of each
topic in each documents. For example, based on the proportion
of words assigned to each topic, we might conclude that a news
article is 45\% about Asian markets, 30\% about debt markets,
and 25\% about other topics.

For LDA inference, we use the collapsed Gibbs sampling method of
\cite{griffiths_steyvers_2004}. Gibbs sampling draws samples (approximately) from the
posterior distribution $p(\bz|\bw)$. The procedure iteratively
sweeps through the topic assignments, at each step resampling
the assignment $\bz(i,j)$ conditional on $\bw$ and conditional
on the collection $\bz_{-(i,j)}$ of all assignments other than the
$(i,j)$th. The update probabilities take the form
\begin{equation}
p(\bz(i,j)=k|\bz_{-(i,j)},\bw) \propto \frac{N_{k,v}+\beta}{\sum_{v'}(N_{k,v'}+\beta)}(N_{j,k}+\alpha),\label{gibbs}
\end{equation}
where $v=\bw(i,j)$, $N_{k,v}$ counts the number of times word $v$ is assigned to topic $k$ under the assignments $\bz_{-(i,j)}$, and $N_{j,k}$ counts the number of words in document $j$ assigned to topic $k$ under the assignments $\bz_{-(i,j)}$. These counts exclude $\bz(i,j)$ and $\bw(i,j)$. Through repeated application of (\ref{gibbs}) the distribution of the assignments $\bz$ converges to the posterior distribution $p(\bz|\bw)$.

\subsection{Supervised LDA}
\label{s:slda}

The generative model for sLDA expands the LDA model to include the generation of a label $y(j)$, $j=1,\dots,J$, for each document. In our main application, a document consists of all news articles about a company in a given trading period (either intraday or overnight), and $y(j)$ is that period's stock return for that company.

In sLDA, labels are driven by the distribution of document topics.  For each document $j$, let $\bar{z}_j$ denote the vector of topic frequencies
$$
\bar{z}_j(k) = \frac{1}{N_j}\sum_{i=1}^{N_j}\mathbf{1}\{z(i,j)=k\},
\quad k=1,\dots,K,
$$
so $\bar{z}_j(k)$ is the fraction of words in document $j$
assigned to topic $k$. Here, $\mathbf{1}\{\cdot\}$ denotes
the indicator function of the event $\{\cdot\}$.
For some vector of coefficients $\eta_* = (\eta_1,\dots,\eta_K)^{\top}$
and standard deviation parameter $\sigma_*$, we add to the
LDA generative model the step
\begin{enumerate}
\item[(2)(b)] Generate $y(j)$ from the normal distribution with mean $\eta_*^{\top}\bar{z}_j$ and variance $\sigma^2_*$.
\end{enumerate}

Our goal is to compute an estimate $\hat{\eta}$ of $\eta_*$ along with the topic model. Given a new news article about a company, we could then apply the trained LDA model to estimate topic frequencies $\bar{z}'$ in the new article and forecast the company's contemporaneous trading period stock return as $y' = \hat{\eta}^{\top}\bar{z}'$.  This would be useful when seeing a news article for the first time and being unsure of how the market {\em should} react to that article, as well as in studying whether observed price responses to news are consistent with historical norms.

Gibbs sampling under LDA can be expanded for inference under sLDA. Under the sLDA generative model, the labels $y(j)$ are independent of the text $\bw$ given the topic assignments $\bz$. The posterior of $\bz$ given $\bw$ and $y$ has the form
$$
p(\bz|\bw,y) = \frac{p(\bz,\bw)p(y|\bz,\bw)}{p(\bw,y)}= \frac{p(\bz,\bw)p(y|\bz)}{p(\bw,y)} \propto
p(\bz,\bw)p(y|\bz).
$$
The factor $p(\bz,\bw)$ is determined entirely by the LDA part of the model and is unaffected by the labels. In this formulation, we are treating $\eta_*$ and $\sigma_*$ as unknown parameters rather than imposing a prior distribution.

Using the conditional normality in Step (2)(b), we get
\begin{eqnarray}
p(\bz|\bw,y) 
&\propto& p(\bz,\bw)\prod_{j=1}^{J} \frac{1}{\sqrt{2\pi}\sigma_*}e^{-\frac{1}{2\sigma^2_*}(y(j)-\eta^{\top}_*\bar{z}_j)^2} \\
\label{postrue}
&=& p(\bz,\bw) \prod_{j=1}^J\varphi(y(j);\eta^{\top}_*\bar{z}_j,\sigma^2_*) \nonumber
\end{eqnarray}
with $\varphi(\cdot;\mu,\sigma^2)$ the normal
density with mean $\mu$ and variance $\sigma^2$.

The stochastic EM algorithm combines Gibbs sampling with
regression for inference in sLDA. Given topic assignments
$\bz$, the ``M'' step runs a regression of the labels
$y(j)$, $j=1,\dots,J$, on the frequency vectors $\bar{z}_j$, 
$j=1,\dots,J$, to compute an estimate $\hat{\eta}$.
(We do not include a constant in this regression because
the components of each $\bar{z}_j$ sum to 1.)
Given $\hat{\eta}$, the ``E'' step runs (multiple iterations of)
the Gibbs sampler to update the topic assignments;
but the Gibbs sampler now tilts the probabilities in (\ref{gibbs})
to favor topic assignments that do a better job explaining
labels. Let $\bar{z}_j^k$ denote the topic proportions for
document $j$ when $\bz(i,j)$ is set to $k$ and all other
assignments are unchanged. Then (\ref{gibbs}) is replaced with
\begin{eqnarray}
\lefteqn{p(\bz(i,j)=k|\bz_{-(i,j)},\bw,y) } && \nonumber\\
&\propto& 
\left(\frac{N_{k,v}+\beta}{\sum_{v'}(N_{k,v'}+\beta)}(N_{j,k}+\alpha)\right)
\varphi(y(j);\hat{\eta}^{\top}\bar{z}_j^k,\sigma^2).
\label{sgibbs}
\end{eqnarray}
Compared with (\ref{gibbs}), (\ref{sgibbs}) will give greater
probability to topic $k$ if that assignment increases the
$\varphi$ factor --- in other words, if that assignment
reduces the error $y(j) - \hat{\eta}^{\top}\bar{z}_j^k$
in explaining the label $y(j)$ with the current coefficients
$\hat{\eta}$.  We will return to the important question of
choosing $\sigma$ in (\ref{sgibbs}) in Sections \ref{s:perf} and \ref{s:app}.

\subsection{Foundation in Financial Economics}
\label{s:found}

Before proceeding with our investigation into choosing topics to explain returns, we briefly outline a basis for this investigation in financial economics and provide an underpinning for Step (2)(b), which
posits a linear relation between news topics and returns.

Campbell \cite{campbell_1991} shows that stock returns can be decomposed into changes in investor beliefs about future dividend growth and future discount rates.  In their analysis, the stock return $h_{t+1}$ from day $t$ to day $t+1$ can be approximated as
\begin{equation} \label{eq:c-s}
  h_{t+1} \approx a + (E_{t+1} - E_t) \sum_{j=0}^\infty \rho^j \Delta d_{t+1+j} - (E_{t+1} - E_t)
  \sum_{j=1}^\infty \rho^j h_{t+1+j},
\end{equation}
where $a$ is a constant, $\rho<1$ is a constant discounting factor, $d_t$ is the time $t$ log
dividend, $E_t$ is an expectation taken over the investor's information set at time $t$, and the
change-in-beliefs operator $(E_{t+1} - E_t) X$ is shorthand for $E_{t+1} X - E_t X$.  

We conjecture that the information in news articles affects investor beliefs about future discount rates and future dividend growth. In particular, we assume that, for firm $f$, explainable period-over-period changes in investor beliefs about future discount rates $h$ and future dividend growth $\Delta d$ are both linear functions of trading period's $t+1$'s average document-topic assignments for $f$, i.e. $\bar z_{f,t+1}$. This represents the average document-topic assignments of all news articles mentioning $f$ that came out from the end of trading period $t$ to the end of trading period $t+1$.  In addition, we assume the unexplainable change in investor beliefs is normally distributed.  Both of these assumptions are satisfied if, for example, the state variables follow a vector autoregressive process as in \cite{campbell_1991}.  Under these conditions, firm $f$'s returns can be written as
\begin{equation} \label{eq:y} 
h_{t+1} = \eta_*^\top \bar z_{f,t+1} + \epsilon_{t+1} \sim N(\eta_*^\top \bar z_{f,t+1},\sigma_*^2),
\end{equation}
for some vector $\eta$.  The $a$ in (\ref{eq:c-s}) is redundant because the elements of $\bar z$ sum to one. (\ref{eq:y}) is the response equation in our sLDA framework.

\section{News and Stock Market Data}
\label{s:data}

Our text data consists of news articles from the Thomson Reuters (TR) archive between 2015 and 2019.  We select articles that mention firms in S\&P 500 index on the day the article appears. We stem words, remove stop words, drop words that appear less than 7 times total, and drop articles with fewer than 50 words. This leaves a vocabulary of 48,966 words. We exclude articles that mention more than three companies. The process yields 90,544 articles.

If trading period $t+1$ is overnight, then it contains news that came out from 4pm of the prior trading day to 9:30am of the present day.  If it is intraday, the trading period contains news articles from 9:30am to 4pm of the present day.  We count any news over the weekend or holiday as belonging to the overnight trading period from the prior market close (4pm) to the next market open (9:30am).  All times are New York times.  We separate these time periods because this study is part of a larger project to understand differences between intraday and overnight news and returns.

In each trading period we combine all news articles that mention a single company into a single document and label that document with the stock return (or squared return) for that company. For articles mentioning multiple companies, we label the article with the average stock return of the underlying companies. We obtain stock returns from CRSP. 

Some of our analysis also assigns a sentiment score to each article.  We calculate sentiment as the difference between the number of positive and negative words in the article divided by the total number of words. We use the Loughran-McDonald \cite{Loughran-McDonald} lexicon to determine positive and negative words because it is tailored to financial news. Summary statistics for our trading period labels are: returns have a mean of -1.6e-05 and a standard deviation (SD) of 0.0228; squared returns have a mean of 0.0005 and a SD of 0.0042; and for sentiment the mean is -0.013 and the SD is 0.0231.

\section{Performance of sLDA}
\label{s:perf}

The parameter $\sigma$ in the Gibbs step (\ref{sgibbs}) can naturally be thought of as an estimate of the true standard deviation $\sigma_*$ of the residuals $y(j) - \eta^{\top}_*\bar{z}_j$. But it can also be viewed as a mechanism for controlling the relative weight put on the text data $\bw$ and the labels $y$ in choosing the topic assignments $\bz$, with smaller values of $\sigma$ putting greater weight on the labels. Indeed, as $\sigma\to 0$, $\varphi(x;\mu,\sigma^2)$ becomes infinite at $x=\mu$ and zero everywhere else. 

This relative weighting role for $\sigma$ is closely related to the Power-sLDA method of \cite{zhang_kjellstrom_2014} and the prediction constrained method of \cite{hughes_ea_2018}, both of which seek to find topics that improve label predictability. (Those methods do not use Gibbs sampling.) Power-sLDA effectively allows different documents to use different values of $\sigma$, with smaller values for longer documents. PC-sLDA puts a lower bound on the (in-sample) predictability of labels from topic frequencies. When translated to our setting, this corresponds to choosing $\sigma$ through the Lagrange multiplier of the lower-bound constraint, so requiring greater in-sample predictability leads to a smaller $\sigma$.

In our setting, choosing $\sigma$ small to improve label-predictability can give disastrously bad results. A small $\sigma$ can distort the topic model to overfit the labels, rendering the model useless for out-of-sample predictions. Interestingly, \cite{boydgraber_resnik_2010} report unstable results in trying to tune $\sigma$ and fix it at 0.25. 

In Section \ref{s:app}, we explain what happens theoretically as $\sigma\to 0$. But the results of Table \ref{t:r2s} provide a simple illustration.
We run sLDA with $\sigma$ fixed and report the in-sample $R^2$ after the final ''M" step. Using the estimated coefficients $\hat{\eta}$, we also calculate an out-of-sample (predictive) $R^2$ on a holdout sample. (See Section \ref{s:branch} for details of this calculation.) At small values of $\sigma$, the in-sample $R^2$ approaches 1, suggesting that the topic model is doing an unreasonably good job explaining returns. But the predictive $R^2$ becomes very negative, revealing serious overfitting to the training data. In choosing topic assignments to fit the training labels at the expense of the validity of the topic model, the algorithm loses any ability to predict returns out of sample. The problem becomes particularly acute with a large number of topics. With more topics, we are more likely to encounter the scenario described in (\ref{minmax}), below, and more likely to find meaningless topic assignments that produce nearly perfect in-sample fits to the data. 

\begin{table}[ht]
\centering
\caption{$R^2$ using sLDA with fixed $\sigma$}
\label{t:r2s}
\scalebox{0.94}{
\begin{tabular}{|c|*{2}{| >{}c<{}}|*{2}{|>{}c<{}}|}
 \hhline{|-||--||--|}
 \multirow{2}{*}{$\boldsymbol{\sigma}$} &\multicolumn{2}{c||}{\textbf{10 Topics}} &\multicolumn{2}{c|}{\textbf{200 Topics}} \\ 
 \hhline{|~||--||--}
 & In-sample & Out-sample & In-sample & Out-sample \\
 \hhline{=::==::==}
0.9 & 0.0013 & -0.0005 & 0.019 & 0.015 \\
0.75 & 0.001 & -0.0008 &0.0081 & 0.0032 \\
0.5 & 0.001 & -0.0007 & 0.009 & 0.0017 \\
0.25 & 0.001 & -0.0007 &0.019 & 0.018 \\
0.1 & 0.0009 & -0.0007 & 0.012& 0.007 \\
0.01 & 0.0009 & -0.0012 & 0.014 & 0.009 \\
0.001 & 0.0012 & -0.001 & 0.028 & 0.017 \\
1e-04 & 0.0009 & -0.0002 & 0.41 & -0.332 \\
1e-05 & 0.0094 & -0.0062 & 0.95 & -0.441 \\
1e-06 & 0.96 & -5.16 & 0.98 & -0.17 \\
 \hhline{-||--||--}
\end{tabular}}
\end{table}

Choosing a larger or more accurate value of $\sigma$ can help but does not eliminate the problem. If $\sigma$ is too large, sLDA reduces to LDA. The true value $\sigma_*$ in Step (2)(b) of the generative model is the standard deviation of the residuals $y(j) - \eta^{\top}_*\bar{z}_j$, $j=1,\dots,J$. It therefore seems reasonable to estimate $\sigma_*$ from the residuals in the regression in each ``M'' step. Figure \ref{f:updated} shows the evolution of in-sample $R^2$ across independent runs of this method. At a fixed number of iterations, the results vary widely, and along a single path the in-sample $R^2$ sometimes climbs rapidly and then falls.  We study this case in greater detail in Section \ref{s:app}, and here we describe less formally why it often fails.

\def\wscale{0.85}
\begin{figure}[ht]
\caption{sLDA using updated $\sigma$}
\label{f:updated}
\includegraphics[width=\wscale\linewidth]{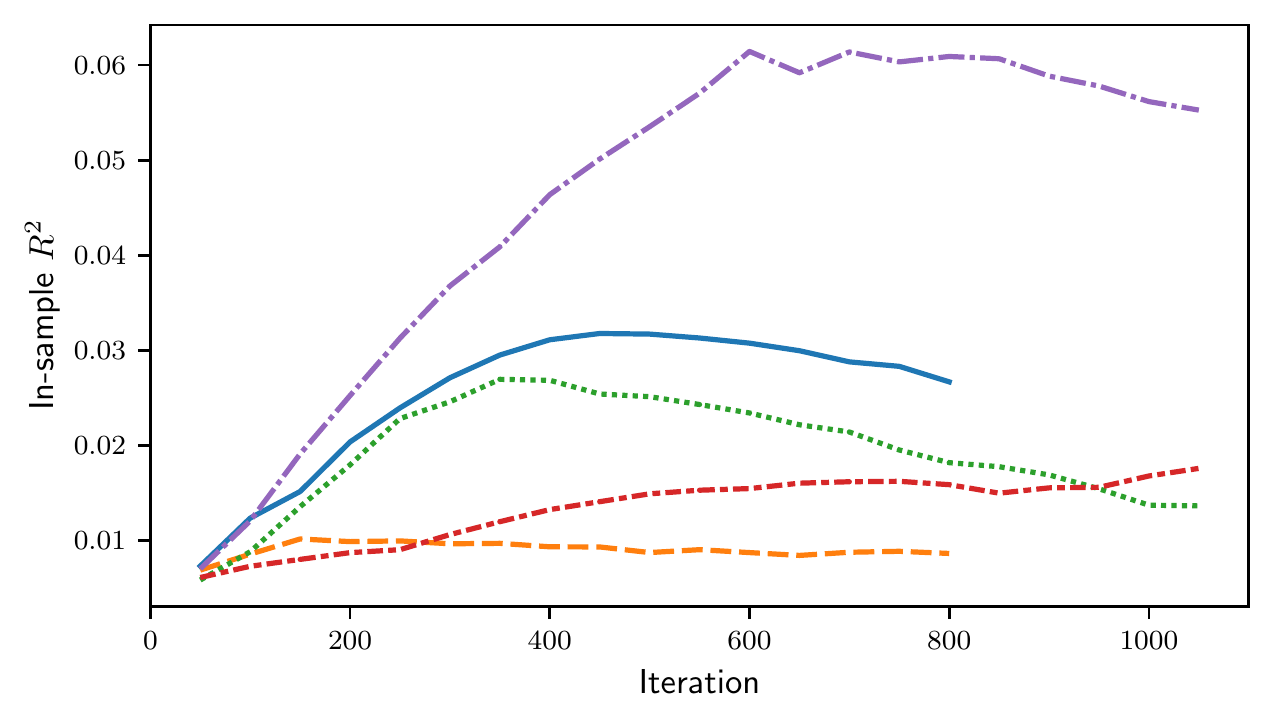}
\end{figure}

From a regression perspective, the root of the problem with the Gibbs/EM algorithm is that it uses the same data to estimate the coefficients $\eta$ and to select the regressors $\bar{z}_j$. This offers too much flexibility and can lead to serious overfitting. Any error in the current estimate of $\hat{\eta}$ gets amplified when the Gibbs sampler favors topic assignments tailored to $\hat{\eta}$; at the next ``M'' step, the errors are further amplified when the new $\hat{\eta}$ is optimized for flawed topic assignments, and the process repeats.

For example, Section \ref{s:app} considers the case that, for fixed $\hat{\eta}=\eta$, 
\begin{equation}
\min_k \eta_k < y(j) < \max_k \eta_k, \quad j=1,\dots,J,
\label{minmax}
\end{equation}
or, through a relabeling of topics,
\begin{equation}
\eta_1 < y(j) < \eta_2, \quad j=1,\dots,J.
\label{ycx}
\end{equation}
Then each $y(j)$ can be expressed as a convex combination
$
y(j) = \lambda_j\eta_1 + (1-\lambda_j)\eta_2 \text{ for } \lambda_j\in(0,1).
$
For a topic assignment $\bz$ with $\bar{z}_j(1)\approx \lambda_j$, $\bar{z}_j(2) = 1-\lambda_j$, and all other $\bar{z}_j(k)=0$, we have
\begin{equation}
y(j) \approx \bar{z}_j(1)\eta_1 + (1-\bar{z}_j(2))\eta_2 = \eta^{\top}\bar{z}_j.
\label{ymix}
\end{equation}
In other words, if at some stage in our EM iterations we get a vector of coefficients that satisfy (\ref{ycx}), then a degenerate set of topic assignments as in (\ref{ymix}) will lead to near-perfect in-sample label fits, but these assignments largely ignore information in the text data.

The Gibbs/EM algorithm, therefore, has two types of overfitting:
\begin{itemize}
\item \emph{M-overfitting}: The usual tendency of regression
coefficients $\hat{\eta}$ to provide better estimates in-sample than out-of-sample;
\item \emph{E-overfitting}: The less familiar result of adapting
the regressors $\bar{z}_j$ to improve fits to labels with $\hat{\eta}$
fixed. This is problematic even if $\hat{\eta}=\eta_*$, the true
vector of coefficients.
\end{itemize}
To address these shortcomings, we need to separate the estimation
of coefficients from the estimation of topic assignments.

\section{Branching LDA}
\label{s:branch}

We achieve this separation through a more computationally intensive process that reinforces topic assignments that perform better out-of-sample.  We compare a few alternatives. All methods split the data into training, validation, and test (holdout) samples.

\textbf{LDA.} As a baseline, we run ordinary collapsed Gibbs sampling for LDA for 700 iterations on the training and validation data. We regress the labels $y(j)$ on the final topic assignments $\bar{z}_j$ to compute estimated coefficients $\hat{\eta}$, using only the training and validation data. We then apply the trained LDA model to assign topics $\bz'$ to the holdout sample and compute topic proportions $\bar{z}_j'$ for the held out documents. Using the coefficient vector $\hat{\eta}$ from the training data, we calculate estimates $\hat{\eta}^{\top}\bar{z}_j'$ of the labels $y'(j)$ in the holdout sample. We evaluate performance using the predictive $R^2$, $1 - \frac{\sum_j (y'(j)-\hat{\eta}^{\top}\bar{z}_j')^2}{\sum_j(y'(j)-\bar{y}')^2}$, where $\bar{y}'$ is the average label in the holdout sample, and the sums run over documents in the holdout sample. Importantly, we do not re-estimate $\hat{\eta}$ through a regression on the holdout sample. We repeat this process 10 times. We use 10 replications here and below to see a range of outcomes with reasonable computational effort.

\textbf{Best-Of LDA.} We start with plain LDA. After an initial period of 250 iterations, we keep 10 subruns spaced 50 iterations apart. For each subrun, we evaluate a predictive $R^2$ on the validation sample, using cross-validation: we randomly split the validation sample into five subsets; we estimate $\hat{\eta}$ on four subsets; evaluate the predictive $R^2$ on the fifth subset; and average the predictive $R^2$ over the five ways of doing this split. We repeat this process for 10 independent runs. This gives us 10 topic assignments after 250 iterations, 10 after another 50 iterations, and so on. At each number of iterations, we pick the best model as measured by the average predictive $R^2$ on the validation sample. (We have also experimented with keeping the second best model as a form of regularization.) This leaves us with 10 winning models. We evaluate their predictive $R^2$s on the holdout sample using the procedure described for plain LDA. 

\textbf{Branching LDA.} We start with 10 independent plain LDA runs, and, after an initial period of 250 iterations, use the cross-validation procedure described under Best-Of LDA to compute a predictive $R^2$ for each of these.  The run with the highest predictive $R^2$ becomes the parent run.  We then branch the Gibbs sampler into 10 new independent ``child'' runs.  Each child run starts with an identical topic assignment $\bz$ to the parent.  After 50 iterations, we again use the cross-validation procedure described under Best-Of LDA to compute a predictive $R^2$ for the 10 child runs and the parent run. The best of these eleven runs becomes the new parent. (We have also experimented with keeping the second best subrun and excluding the previous parent as forms of regularization.)  We continue the Gibbs sampling with 10 independent child runs for 50 iterations, and repeat the process 8 times. We evaluate the parents from the last 10 branchings based on their predictive $R^2$ on the holdout sample. 

Best-Of LDA and Branching LDA are two of many possible ways to address the two types of overfitting described at the end of Section \ref{s:perf}. We never estimate $\eta$ on the same data we use to evaluate predictions, so we never commit M-overfitting. More importantly, we never use the current $\eta$ estimate to steer the Gibbs sampler, so we never commit E-overfitting; we instead use $\eta$ for cross-validation. Nevertheless, compared with plain LDA, these methods reinforce outcomes of the LDA Gibbs sampler that perform best at estimating labels in the validation set, and we will see that this translates to improved performance on the holdout set. Of the many realizations of the topic assignments $\bz$ that are roughly equally good at explaining the text $\bw$, we are searching for those that do a better job explaining the labels $y$.

\section{Numerical Results}
\label{s:num}

We compare results with four types of labels. As a benchmark, we use data simulated under the sLDA generative model; this allows us to compare methods when the model is correctly specified. We test 100- and 200-topic models using $V=$ 5,000, $J=$  90,000, and $N_j=$ 524, which is the average number of words per document in our TR data. Our simulated labels have an $R^2$ of 0.25 when regressed on topic frequencies. With the TR data, we let $y(j)$ be either the stock return for each company in each period, or the squared return (volatility), or the same-trading-period news sentiment.  The returns and volatilities are contemporaneous with the topic assignments.  We include sentiment because it has an $R^2$ of around 0.5, whereas the $R^2$s for returns and volatility are very small.

\begin{figure}[ht]
    \caption{Out-of-sample $R^2$s; x-axis gives the number of topics.} \label{f:oosR2}
    \includegraphics[width=\wscale\linewidth]{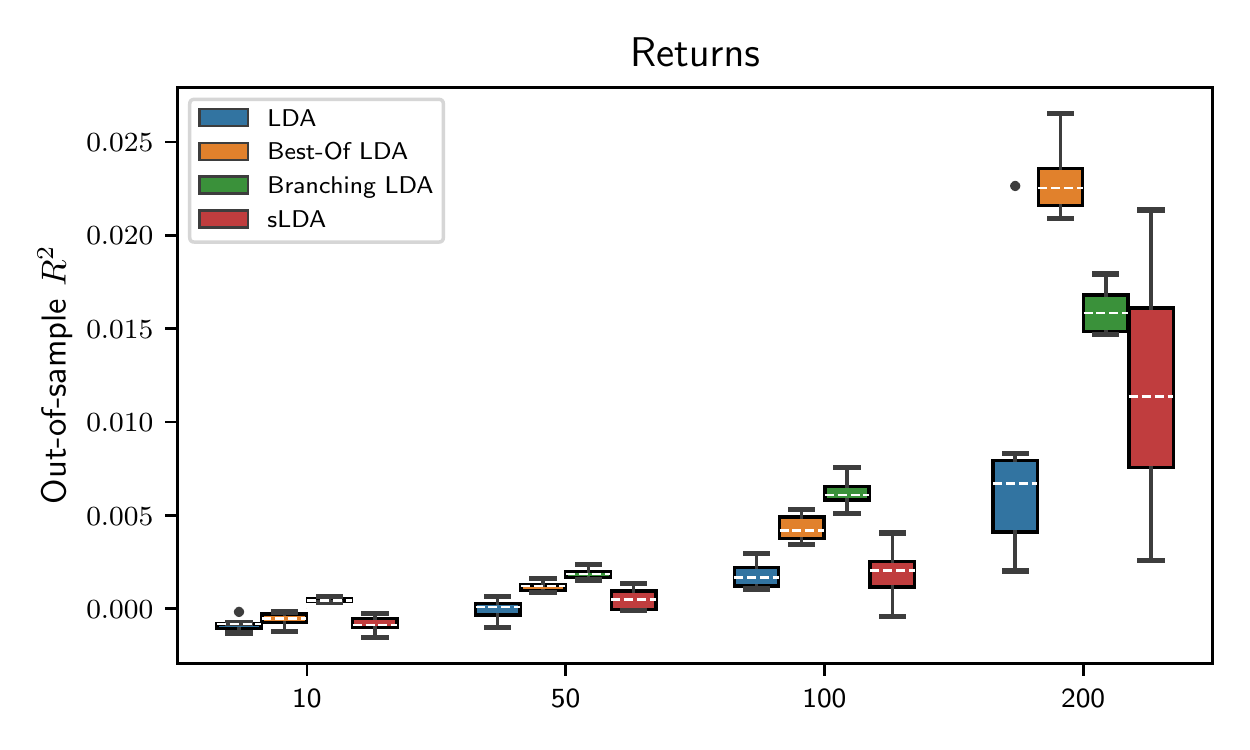}
    \includegraphics[width=\wscale\linewidth]{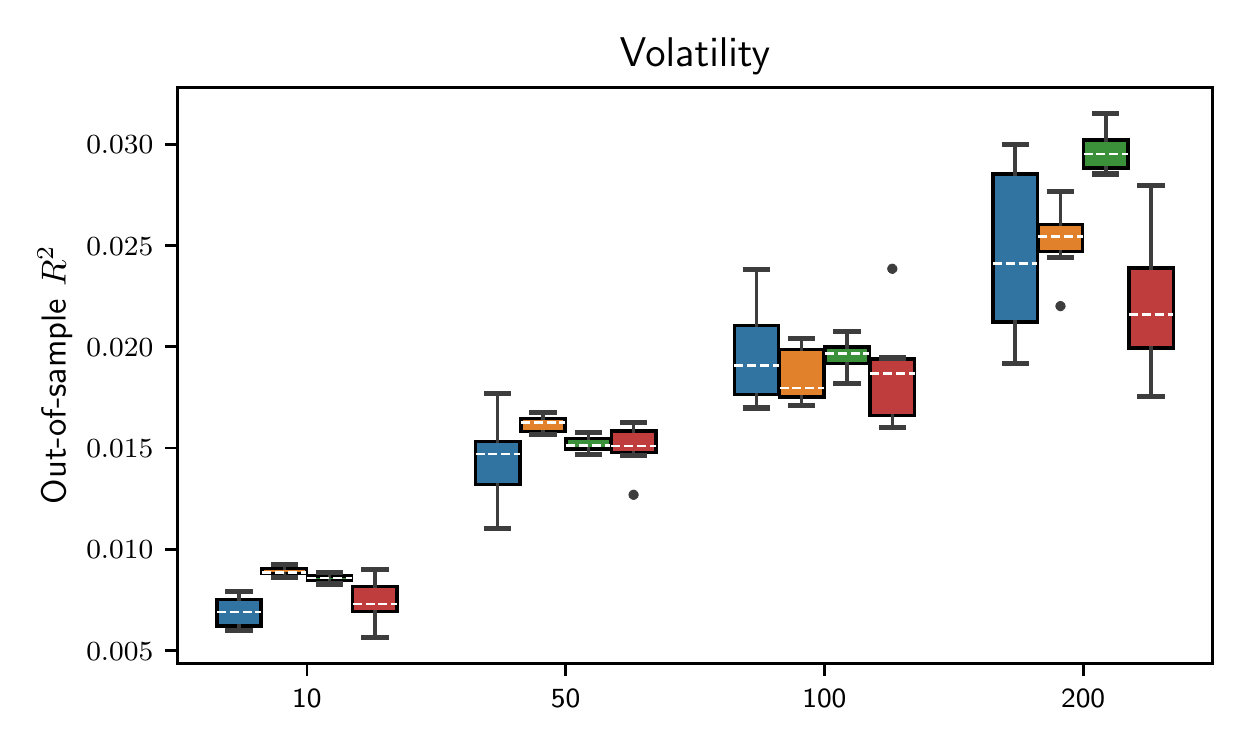}
    \includegraphics[width=\wscale\linewidth]{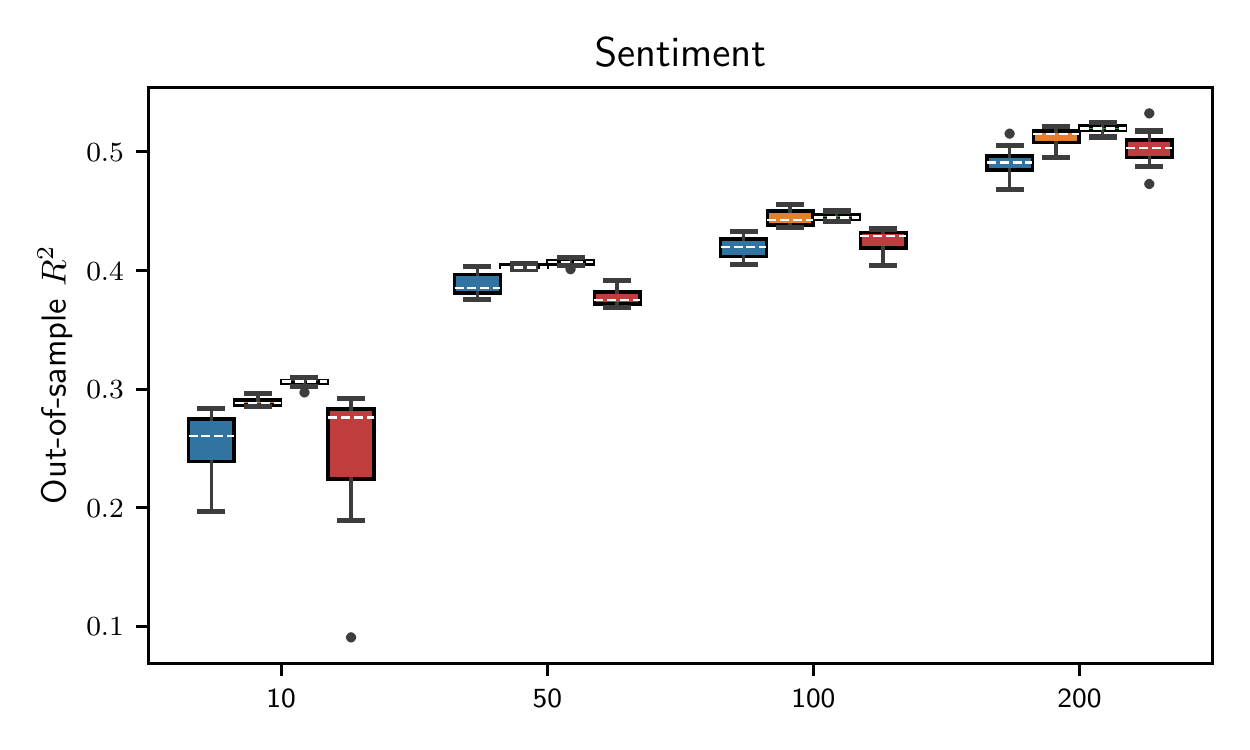}
\end{figure}

For each of these four types of labels, we compare four methods (LDA, Best-Of LDA, Branching LDA, and sLDA) through box plots of 10 predictive $R^2$s for each method. (We use 10 to see a range of outcomes with reasonable computational effort.)  We repeat the comparison for 10, 50, 100, and 200 topics, except that the simulated data uses only 100 and 200.  For sLDA, we fix the $\sigma$ parameter at the standard deviation of the labels $y(j)$. Using a fixed value reduces the risk of overfitting, and using the standard deviation of the labels (rather than an arbitrary value like 0.25) ensures that $\sigma$ has roughly the right magnitude.

\begin{figure}[ht] \ContinuedFloat
    \caption{continued; x-axis gives number of topics.}
    \includegraphics[width=\wscale\linewidth,trim={0 0.71cm 0 0},clip]{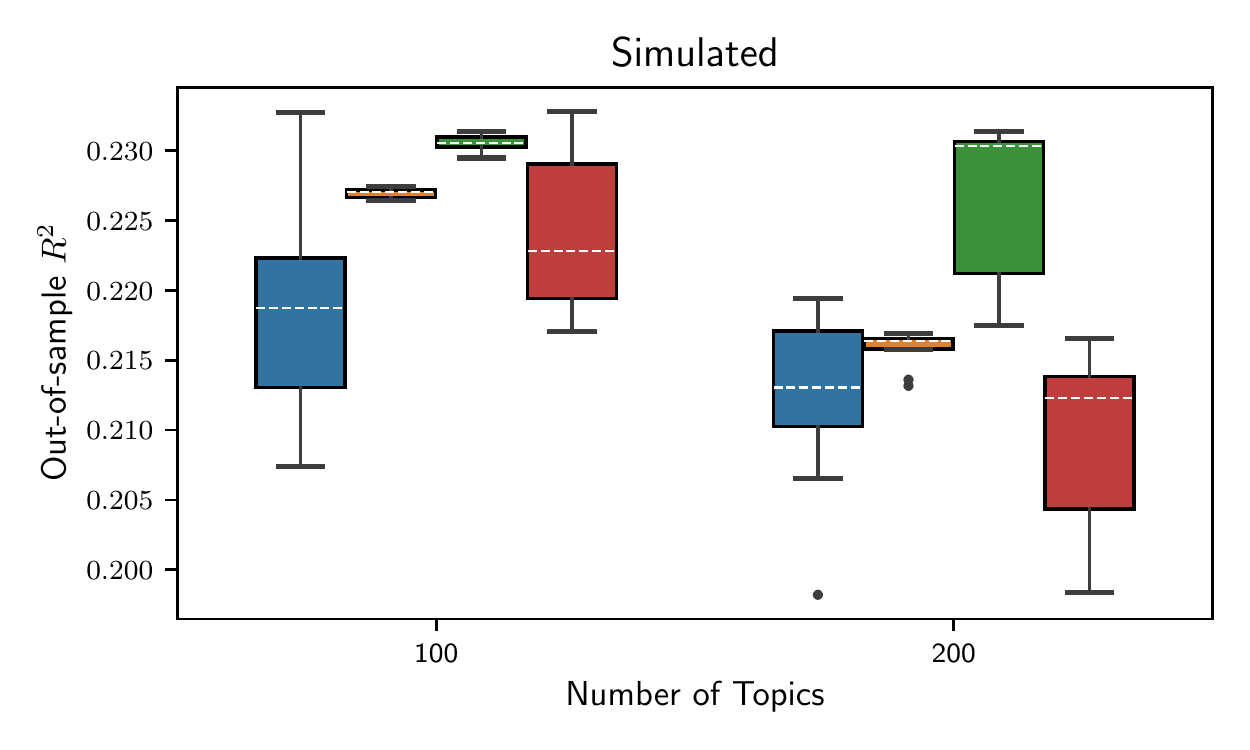}
\end{figure}
\noindent

As can be seen from Figure \ref{f:oosR2}, in all cases, the highest median performance (shown by a dashed white line in the box plots) is produced by either Best-Of or Branching LDA. This consistent pattern confirms the value of our overall approach to searching for topic models while avoiding overfitting.

In explaining returns (top panel), Branching LDA is the consistent winner at 10, 50, and 100 topics. In all comparisons, LDA and sLDA show the greatest dispersion in results. Best-Of and Branching produce more consistent results.
In the context of the in-sample 0.45\% $R^2$ for contemporaneous returns in \cite{tsm-2008}, our 100- and 200-topic model results are impressive.  It should be noted that our $R^2$'s are out-of-sample since we don't re-estimate the $\hat \eta$ vector in the holdout sample; this makes our results that much more noteworthy.

Finding topics to explain sentiment (third panel) should be an easier case for sLDA because sentiment has a much higher $R^2$. But Branching LDA performs a bit better at all four numbers of topics.

Figure \ref{f:ll} compares log-likelihoods across experiments for returns. These log-likelihoods are calculated from the training data for the text data only, ignoring the labels, following the standard procedure used for plain LDA. They therefore provide measures of the quality of fit to the text data. Two features stand out: at a given number of topics, there is very little variation within each method, and there is very little variation across methods. The methods perform roughly equally well in fitting the training text. Branching LDA achieves a slightly lower log-likelihood, but this is entirely appropriate given its higher predictive $R^2$. By reinforcing the fit to the labels, we expect to give up some fit to the text.

As a check on the topics identified by the various methods, we show examples in Table \ref{t:topics}. For each method, we pick the run that produced the highest predictive $R^2$ for returns in a 200-topic model. For each such winning run, we find the three most negative and most positive topics, as measured by the coefficient estimates $\hat{\eta}_k$. For each of these topics, we show the top 10 most probable words.

All of the topics displayed look sensible, and a couple show up consistently across methods. Of the LDA topics, the third negative and second positive topic are a bit surprising as candidates for most influential topics, as is the second most negative topic under sLDA.  In contrast, Best-Of LDA and Branching LDA find interesting candidates for their most positive topics. These comparisons are anecdotal, but they support the quality of the text portion of models designed to explain returns.

\begin{figure}[ht]
\caption{Comparisons of log-likelihood for return models} \label{f:ll}
\includegraphics[trim={0 0.1cm 0 0},clip,width=\wscale\linewidth]{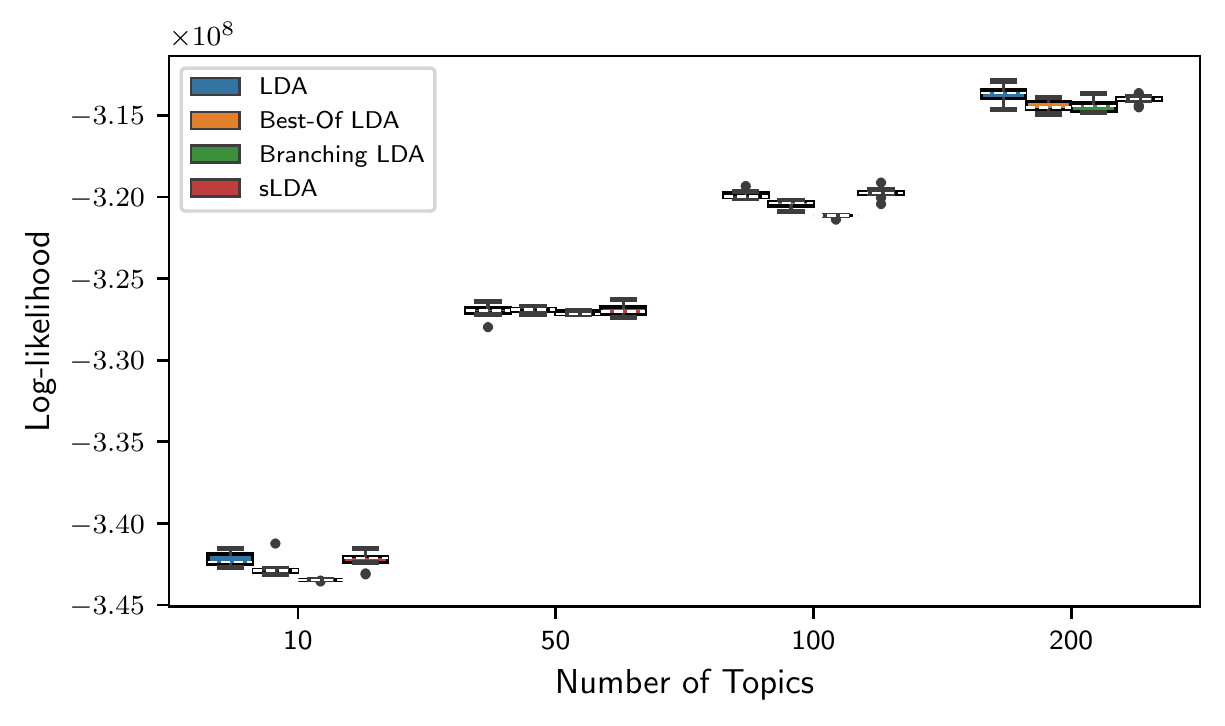}
\end{figure}

\begin{table*}[ht]
\caption{Most-negative/positive topics and their coefficients in explaining returns with 200 topics}
\label{t:topics}
\centering
\scalebox{0.85}{
\begin{tabular}{|l|lc|c|}
\hline
 \multicolumn{3}{|l|}{\textbf{LDA}} & Coefficient  \\
\hline
 Negative   & 1 &  price	cut	drop	fall	fell	lower	declin	analyst	quarter	weak & -0.098 \\
 \cline{2-4}
  & 2 & insid	https	video	watch	morn	trade	reut	short	transcript	open & -0.011 \\
  \cline{2-4}
 & 3 & climat	coal	carbon	emiss	energi	cruis	gas	fuel	environment	norway & -0.009 \\
 \hline
 Positive & 1  & quarter	revenu	analyst	expect	earn	estim	profit	result	forecast	rose & 0.015 \\
 \cline{2-4}
  & 2  & korea	korean	hospit	hca	seoul	tenet	oper	kim	jin	lee & 0.011 \\
  \cline{2-4}
  & 3  & deal	offer	buy	acquisit	cash	close	merger	bid	combin	agre & 0.011 \\
\hline\hline
 \multicolumn{4}{|l|}{\textbf{Best-Of LDA}}  \\
\hline
 Negative   & 1 &  analyst	fall	stock	drop	cut	fell	forecast	price	declin	expect & -0.090 \\
 \cline{2-4}
  & 2 & bankruptci	debt	file	creditor	apollo	restructur	oper	fund	protect	manag & -0.013 \\
  \cline{2-4}
 & 3 & store	sale	retail	maci	toy	depart	brand	holiday	subrat	patnaik & -0.008 \\
 \hline
 Positive & 1  & fund	hedg	manag	investor	activist	capit	elliott	ackman	stock	invest & 0.014 \\
 \cline{2-4}
  & 2  & quarter	revenu	analyst	expect	profit	estim	rose	earn	forecast	net & 0.011\\
  \cline{2-4}
  & 3  & deal	offer	buy	acquisit	close	cash	bid	combin	merger	agre & 0.011 \\
\hline\hline
 \multicolumn{4}{|l|}{\textbf{Branching LDA}} \\
\hline
 Negative   & 1 &  quarter	expect	cut	profit	price	analyst	result	forecast	lower	drop & -0.063 \\
 \cline{2-4}
  & 2 & store	retail	sale	dollar	sruthi	maci	home	ramakrishnan	general	tree  & -0.008 \\
  \cline{2-4}
 & 3 & investig letter	inform	probe	statement	review	request	issu	regul	alleg & -0.008 \\
 \hline
 Positive & 1  & board	activist	investor	icahn	sharehold	elliott	fund	manag	stake	director & 0.019 \\
 \cline{2-4}
  & 2  & quarter	revenu	profit	analyst	earn	expect	cent	net	estim	rose & 0.011 \\
  \cline{2-4}
  & 3  & deal	buy	acquisit	close	merger	combin	agre	transact	acquir	cash & 0.010 \\
\hline\hline
 \multicolumn{4}{|l|}{\textbf{sLDA}}  \\
\hline
 Negative   & 1 &  quarter fell drop fall	analyst expect declin forecast cut sale & -0.139 \\
 \cline{2-4}
  & 2 & solar energi power wind project renew electr sunedison panel develop & -0.007 \\
  \cline{2-4}
 & 3 & ipo offer public price stock initi rang	 rais lanc net & -0.006 \\
 \hline
 Positive & 1  & growth sale profit quarter	expect revenu forecast result margin	analyst & 0.016 \\
 \cline{2-4}
  & 2  & quarter revenu	 analyst cent expect profit estim  earn rose net & 0.014 \\
  \cline{2-4}
  & 3  & deal buy acquisit merger close combin busi acquir cash agre & 0.012 \\
\hline
\end{tabular}}
\end{table*}

\section{Analysis of the \texorpdfstring{$\sigma$}{sigma} Parameter}
\label{s:app}

This section formalizes the observations in Section \ref{s:perf} on E-overfitting in sLDA.  We examine what happens when the $\sigma$ parameter is too small or when it is estimated through regression residuals.

\subsection{Small \texorpdfstring{$\sigma$}{sigma}}

Building on (\ref{postrue}), consider approximate posterior distributions
\begin{equation}
p_a(\bz|\bw,y;\eta,\sigma) \propto p(\bz,\bw)\prod_{j=1}^{J} e^{-\frac{1}{2\sigma^2}(y_j-\eta^{\top}\bar{z}_j)^2},
\label{posapp}
\end{equation}
with $\sigma>0$ and $\eta\in\mathbb{R}^K$.
We use the notation $p_a$ to emphasize that (\ref{posapp}) is only an
approximation to the true posterior in (\ref{postrue}).
We always use an approximation in practice because $\sigma_*$
and $\eta_*$ are unknown.

For fixed $\eta$, let
$
\Delta = \min_{\bz'} \sum_{j=1}^{J}(y_j - \eta^{\top}\bar{z}'_j)^2,
$
the minimum taken over all possible assignments of the $K$ topics
to the words in the corpus. $\Delta$ denotes the minimum mean squared error (MSE)
achievable for fixed $\eta$. Denote by
$
Z_{\Delta} = \{\bz': \sum_{j=1}^J(y_j - \eta^{\top}\bar{z}'_j)^2=\Delta\}
$
the set of topic assignments that achieve the minimum MSE for
$\eta$.

The value of $\Delta$ at $\eta$ could be larger
or smaller than the value of $\Delta$ at the true coefficient vector $\eta_*$,
so $\Delta$ is not a reliable measure of the ability of topics to predict
the labels $y$; and,
for $\eta\not=\eta_*$, membership of $\bz'$ in $Z_{\Delta}$ is not
a reliable measure of the ability of the topic assignments $\bz'$ to explain
the data $y$. In particular, for $\eta\not=\eta_*$, there is no reason
to expect that true or representative topic assignments are in $Z_{\Delta}$.

Despite the fact that the topic assignments in $Z_{\Delta}$ are not
particularly meaningful, we will see that the approximate posterior (\ref{posapp})
concentrates probability on this set when $\sigma$ is small, 
and largely ignores the information $p(\bz|\bw)$
from the text data. Without an estimate of the true parameter $\sigma_*$,
one may inadvertently choose $\sigma$ too small. More importantly, the
methods of \cite{zhang_kjellstrom_2014} and \cite{hughes_ea_2018} indirectly drive $\sigma$
to smaller values by raising the Gaussian factor in the posterior to a power greater
than 1.

For $\bz\in Z_{\Delta}$, let
$q(\bz,\bw) = p(\bz,\bw)/\sum_{\bz'\in Z_{\Delta}} p(\bz',\bw)$.

\begin{proposition}
\label{p:p1}
As $\sigma\to 0$, the approximate posterior (\ref{posapp}) becomes concentrated
on $Z_{\Delta}$:
$$
\lim_{\sigma\to 0}p_a(\bz|\bw,y;\eta,\sigma) = 
\left\{\begin{array}{ll}
0, & \bz\not\in Z_{\Delta}; \\
q(\bz,\bw), & \bz\in Z_{\Delta}.
\end{array}\right.
$$
\end{proposition}

\begin{proof}
The right side of (\ref{posapp}) is
\begin{eqnarray}
\lefteqn{
\frac{p(\bz,\bw)\exp\{-\frac{1}{2\sigma^2}\sum_j(y_j-\eta^{\top}\bar{z}_j)^2\}}{\sum_{\bz'}p(\bz',\bw)\exp\{-\frac{1}{2\sigma^2}\sum_j(y_j-\eta^{\top}\bar{z}'_j)^2\}} 
} && \nonumber \\
&=&
\frac{p(\bz,\bw)\exp\{-\frac{1}{2\sigma^2}(\sum_j(y_j-\eta^{\top}\bar{z}_j)^2-\Delta)\}}{\sum_{\bz'}p(\bz',\bw)\exp\{-\frac{1}{2\sigma^2}(\sum_j(y_j-\eta^{\top}\bar{z}'_j)^2-\Delta)\}}
\label{prat} 
\end{eqnarray}
For $\bz'\not\in Z_\Delta$,
$\sum_j(y_j-\eta\bar{z}'_j)^2-\Delta >0$,
so, as $\sigma\to0$,
$$
-\frac{1}{2\sigma^2}\left(\sum_j(y_j-\eta\bar{z}'_j)^2-\Delta\right) \to -\infty.
$$
The denominator of (\ref{prat}) converges to $\sum_{\bz'\in Z_{\Delta}}p(\bz',\bw)$; the numerator equals $p(\bz,\bw)$ if $\bz\in Z_{\Delta}$ and converges to zero otherwise. Thus, in the limit as $\sigma\to0$, $p_a(\bz|\bw,y;\eta,\sigma)$, is proportional to $p(\bz,\bw)$ for $\bz\in Z_{\Delta}$ and zero otherwise. As $p_a$ is (\ref{prat}) normalized so that its sum over $\bz$ is 1, the result follows.
\end{proof}

This result shows that if $\sigma$ is small then the approximate posterior
concentrates probability on the topic assignments that approximate
the labels $y$ well (the topic assignments in $Z_{\Delta}$), which
may not be particularly meaningful, especially if $\eta\not=\eta_*$. 
The topic model (through $p(\bz,\bw)$) plays a secondary role by determining the
relative weight of the assignments $\bz\in Z_{\Delta}$.

\subsection{Residual \texorpdfstring{$\sigma$}{sigma}}

The breakdown in Proposition~\ref{p:p1} results from using an arbitrary
value for $\sigma$. 
An alternative is to use the standard regression estimate,
\begin{equation}
\hat{\sigma} = \sqrt{\frac{1}{J-K}\sum_{j=1}^{J}\epsilon^2_j},
\quad \epsilon_j = y_j - \eta^{\top}\bar{z}_j,
\label{sigeps}
\end{equation}
where, as before, $J$ is the number of documents and $K<J$ is the number
topics. 
For any coefficient vector $\eta\in\mathbb{R}^K$, consider the approximate posterior
\begin{equation}
p_a(\bz|\bw,y;\eta) \propto p(\bz,\bw)\prod_{j=1}^J \frac{1}{\sqrt{2\pi}\hat{\sigma}}e^{-\frac{1}{2\hat{\sigma}^2}(y_j-\eta^{\top}\bar{z}_j)^2}.
\label{peta}
\end{equation}

Suppose (\ref{ycx}) holds.
In a document with $N$ words, the components of $\bar{z}_j$ are
multiples of $1/N$, so $\lambda_j$ can be approximated by
$\bar{z}_j(1)$ to within $1/2N$, yielding the bound
\begin{equation}
|y_j -  \eta^{\top}\bar{z}_j| < \frac{\eta_2-\eta_1}{N},
\quad j=1,\dots,J.
\label{resbnd}
\end{equation}
Assume for simplicity that all documents have $N$ words.
It follows that this choice of $\bz$ makes
\begin{equation}
\hat{\sigma} \le \sqrt{\frac{J}{J-K}}\left(\frac{\eta_2-\eta_1}{N}\right).
\label{sigsmall}
\end{equation}
With a large number of words per document $N$, we can achieve
a near-perfect fit to the labels $y_j$ through a ``degenerate'' assignment
of words to just topics 1 and 2, as in (\ref{ymix}), ignoring the text data.

In contrast to (\ref{sigsmall}), we do {\it not\/} expect the standard deviation
of the true errors (under the sLDA generative model) to vanish as $N$ increases. Nevertheless, we will show that the approximate
posterior (\ref{peta}), based on the
regression standard error estimate (\ref{sigeps}), overweights degenerate topic assignments
achieving (\ref{ymix}) and (\ref{sigsmall}).

To make this idea precise, we need to consider a sequence of problems through which the number of words per document $N$ and the number of documents $J$ increase, holding fixed the number of topics $K$ and the vocabulary size $V$. Within the LDA generative model, we can think of $(\bz,\bw)$ as derived from the first $N$ words in the first $J$ documents in an infinite array.  The labels $y(j)$ work slightly differently because $y(j)$ depends on the proportions of topic assignments in document $j$ among the first $N$ words, so all $y(j)$ change as $N$ increases. With $\bar{z}_j$ approaching $\theta_j$ as $N$ increases, this is a minor point.

We want to compare the (approximate) posterior probabilities of two sets of topic assignments $\bz$ and $\bz'$, given $(\bw,y)$. For $\bz$, we have in mind topic assignments that ignore the words 
$\bw$ to achieve the near-perfect fit to labels in (\ref{ymix}). We capture the idea that the assignments $\bz'$ are more consistent with the true topic assignments by assuming that the resulting error standard deviation
\begin{equation}
\hat{\sigma}' = \sqrt{\frac{1}{J-K}\sum_{j=1}^{J} {\epsilon'}^2},
\quad
\epsilon' = y(j) - \eta^{\top}\bar{z}'_j,
\label{sigeps2}
\end{equation}
remains bounded away from zero. (Under the sLDA generative model,
with the true vector $\eta_*$ we expect $\hat{\sigma}'$ to approach
the true standard deviation $\sigma_*$.)
Thus, we expect  ``good'' topic assignments to belong to
$$
Z_G = \{\bz': \inf_{N\ge 1, J\ge K+1}\hat{\sigma}' >0\},
$$
whereas the ``bad'' topic assignments satisfying (\ref{ymix}) belong to
$$
Z_B = \{\bz: \sup_{N\ge 1, J\ge K+1}N \hat{\sigma} <\infty\}.
$$

We use the notation in (\ref{peta}), but we drop $\eta$ as an explicit
parameter on the left to simplify notation.
For any topic assignments $\bz$ and $\bz'$, the ratios
$p_a(\bz|\bw,y)/p_a(\bz'|\bw,y)$ and $p(\bz|\bw)/p(\bz'|\bw)$
measure the relative (approximate) posterior probabilities of $\bz$ and $\bz'$ given $(\bw,y)$ and (exact posterior probabilities) given $\bw$, respectively.  The ratio of ratios
$$
R(\bz,\bz';\bw,y) = 
\left(\frac{p_a(\bz|\bw,y)}{p_a(\bz'|\bw,y)}\right)/\left(\frac{p(\bz|\bw)}{p(\bz'|\bw)}\right)
$$
measures how this relative probability is affected by the labels $y$ compared
with the text $\bw$.  Larger values of $R(\bz,\bz';\bw,y)$ indicate that the approximate
posterior (\ref{peta}) with $y$ included shifts more weight to $\bz$ relative to $\bz'$.
The following result shows that this ratio grows  fast for
$\bz\in Z_B$ and $\bz'\in Z_G$, indicating that (\ref{peta}) overwhelmingly
shifts weight from ``good'' to ``bad'' topic assignments.

\begin{proposition}
\label{p:reg}
If $\bz\in Z_B$ and $\bz'\in Z_G$ with
$\hat{\sigma}$ and $\hat{\sigma}'$ calculated according to
(\ref{sigeps}) and (\ref{sigeps2}), then for any fixed $J\ge 1$,
\begin{equation}
\liminf_{N\to\infty} \frac{\log R(\bz,\bz';\bw,y)}{\log N} > 0,
\label{rlim1}
\end{equation}
and, for all sufficiently large (but fixed) $N$,
\begin{equation}
\liminf_{J\to\infty} \frac{\log R(\bz,\bz';\bw,y)}{J} > 0.
\label{rlim2}
\end{equation}
\end{proposition}
Roughly speaking, the limit in (\ref{rlim1})  says that $R(\bz,\bz';\bw,y)$ grows faster than $N^c$, for some $c>0$, and (\ref{rlim2}) says that it grows faster than $c^{J}$, for some $c>1$. Thus, as either the number of words per document grows or the number of documents grows, (\ref{peta}) puts much more probability on $\bz\in Z_B$ than $\bz'\in Z_G$, relative to what the text model alone would do.

\begin{proof}
From (\ref{peta}) we get
$$
R(\bz,\bz';\bw,y) = \left(\frac{\hat{\sigma}'}{\hat{\sigma}}\right)^{J}
\exp\left(\frac{1}{2\hat{\sigma}^{\prime 2}}\sum_{j=1}^{J}{\epsilon}^{\prime 2}_j-
\frac{1}{2\hat{\sigma}^2}\sum_{j=1}^{J}\epsilon^2_j\right).
$$
In light of (\ref{sigeps}) and (\ref{sigeps2}),
$$
\frac{1}{2\hat{\sigma}^{\prime 2}}\sum_{j=1}^{J}{\epsilon}^{\prime 2}_j-
\frac{1}{2\hat{\sigma}^2}\sum_{j=1}^{J}\epsilon^2_j
= \frac{1}{2}(J-K)-\frac{1}{2}(J-K) = 0, \text{ and so}
$$
\begin{equation}
\log R(\bz,\bz';\bw,y) = J[\log(\hat{\sigma}')-\log(\hat{\sigma})]. 
\label{rlower}
\end{equation}
But $\bz'\in Z_G$
means that $\hat{\sigma}'$ is bounded away from zero, so \\*
$\liminf_{N\to\infty} \log\hat{\sigma}'/\log N \ge 0$;
and $\bz\in Z_B$ implies that $\hat{\sigma} <c/N$, for
some constant $c>0$, so $\limsup_{N\to\infty} \log\hat{\sigma}/\log N \le -1$.
Applying these limits
in (\ref{rlower}) yields (\ref{rlim1}).
For (\ref{rlim2}), observe that because $\hat{\sigma}'$ is bounded away
from zero as $N$ increases and $\hat{\sigma}$ is bounded above
by $c/N$, for some $c>0$, we have $(\hat{\sigma}'/\hat{\sigma})>1+\delta$,
for all sufficiently large $N$, for some $\delta>0$. As $R(\bz,\bz';\bw,y)=(\hat{\sigma}'/\hat{\sigma})^{J}$,
(\ref{rlim2}) follows.
\end{proof}

\section{Conclusions}

We have developed an approach to finding topics in news articles to explain stock returns, volatility, and other types of labels. Supervised topic modeling requires balancing the ability of a model to explain labels and text simultaneously. We have shown that a standard stochastic EM approach to sLDA is vulnerable to serious overfitting. We avoid overfitting through a random search of LDA topic assignments while reinforcing configurations that perform well in cross-validation. Our approach improves over standard methods in tests on a large corpus of business news articles.  The tools developed here are part of a broader investigation into differences between intraday and overnight news and returns.

\bibliographystyle{ACM-Reference-Format}
\bibliography{ICAIF_paper_82}
\end{document}